\documentclass{ifacconf}

\usepackage[normalem]{ulem}
\usepackage{graphicx}      
\usepackage{natbib}        
\usepackage{amsmath,amssymb,amsfonts}
\usepackage{enumitem}
\usepackage{fancyhdr}
\usepackage{graphicx}
\usepackage{setspace}
\usepackage{epstopdf}
\usepackage{color}
\usepackage{booktabs}
\usepackage{multirow}
\usepackage{longtable}
\usepackage{rotating}
\usepackage{colortbl}
\usepackage{threeparttable}
\usepackage{arydshln}
\usepackage{mathtools}
\usepackage{booktabs}
\usepackage{makecell}
\usepackage{xcolor}
\allowdisplaybreaks[2]
\newtheorem{rmk}{Remark}
\newtheorem{assumption}{Assumption}

\newtheorem{proposition}{Proposition}

\newtheorem{definition}{Definition}
\allowdisplaybreaks[2]
\begin{document}
\begin{frontmatter}

\title{Learning-based data-enabled moving horizon estimation with application to membrane-based biological wastewater treatment process\thanksref{footnoteinfo}} 

\thanks[footnoteinfo]{This research is supported by Ministry of Education, Singapore, under its Academic Research Fund Tier 1 (RG95/24 \& RG63/22).\\
$^{1}$ \hspace{0.5mm}Corresponding author: Xunyuan Yin. Tel: (+65) 6316 8746. Email: xunyuan.yin@ntu.edu.sg.\\
}

\author[First]{Xiaojie Li} 
\author[First,Second]{Xunyuan Yin}$^{,1}$
\address[First]{School of Chemistry, Chemical Engineering and Biotechnology, Nanyang Technological University, 62 Nanyang Drive, 637459, Singapore (e-mail: xiaojie002@e.ntu.edu.sg, xunyuan.yin@ntu.edu.sg)}
\address[Second]{Nanyang Environment and Water Research Institute, Nanyang Technological University, 1 CleanTech Loop, 637141, Singapore}

\begin{abstract}                
{\color{black}
In this paper, we propose a data-enabled moving horizon estimation (MHE) approach for a class of nonlinear systems without explicit modeling, by leveraging Koopman operator theory and Willems’ fundamental lemma. Specifically, the nonlinear system is lifted to a linear parameter-varying Koopman surrogate, in which the lifting functions and scheduling mappings are learned directly from data using neural networks. Willems’ fundamental lemma is then employed to construct a trajectory-based representation of the Koopman surrogate, which bypasses the explicit identification of the matrices of the Koopman surrogate. Based on this representation, we formulate a convex data-enabled MHE design, which provides real-time estimates of the Koopman surrogate states, from which the states of the original nonlinear system are reconstructed.}
Sufficient conditions are derived to ensure the stability of the estimation error. The effectiveness of the proposed method is illustrated using a simulated membrane-based biological wastewater treatment process.
\end{abstract}


\end{frontmatter}

\section{Introduction}
Full-state information is critical for implementing various advanced control designs (\cite{allgower2012nonlinear}, \cite{lee2014progress}). However, for complex nonlinear systems, obtaining such information in real time is often difficult or prohibitively expensive (\cite{bourgeois2001line}, \cite{turan2021data}). Nonlinear state estimation offers a practical alternative by reconstructing the full states from limited output measurements (\cite{kottakki2014state}, \cite{shao2010constrained}). Among available approaches, moving horizon estimation (MHE) is one of the well-suited frameworks to nonlinear processes, since it can explicitly handle constraints and nonlinearity by incorporating them into an online optimization problem (\cite{schiller2023lyapunov}, \cite{ji2015robust}).

The performance of MHE typically depends on the availability of an accurate dynamic model of the underlying system. When first-principles models are accessible or when data-driven models are obtained through system identification, MHE can be formulated based on such models (\cite{schiller2023lyapunov}, \cite{ji2015robust}, \cite{yin2023data}, \cite{yan2025self}). 
A promising alternative is to develop data-enabled MHE schemes that bypass explicit modeling and directly construct the estimator from data (\cite{wolff2024robust}).
Willems' fundamental lemma (\cite{willems2005note}), which enables implicit representation of a linear time-invariant (LTI) system using its trajectories, has emerged as a powerful foundation for direct data-enabled control (\cite{berberich2020trajectory}, \cite{morato2024data}) and direct data-enabled state estimation (\cite{wolff2024robust}, \cite{turan2021data}). 

Willems' fundamental lemma (\cite{willems2005note}) has been exploited for state estimation. Within this direct data-enabled framework, a Luenberger observer was developed in \cite{adachi2021dual}. To address the unknown inputs and disturbances, an unknown input observer integrating Willems' fundamental lemma was proposed in \cite{turan2021data}.
A data-enabled MHE approach capable of handling state constraints was proposed in \cite{wolff2024robust} to achieve robustness against measurement noise in offline data.
{\color{black}Despite these advances, the available theoretical results remain restricted to linear systems, because Willems' fundamental lemma (\cite{willems2005note}) is applicable only to LTI systems. To the best of our knowledge, MHE methods based on Willems’ fundamental lemma (\cite{willems2005note}) for the nonlinear systems have not yet been systematically investigated.}

{\color{black}
Motivated by the above gap, we aim to propose a data-enabled MHE approach for a class of nonlinear systems based on Willems’ fundamental lemma (\cite{willems2005note}). 
Specifically, we assume the existence of an exact linear parameter-varying (LPV) Koopman surrogate for the nonlinear systems in a lifted space, together with a reconstruction matrix that maps lifted states back to the original system states. 
By extending Willems' fundamental lemma (\cite{willems2005note}) to the LPV setting, we obtain an implicit trajectory-based representation of the Koopman surrogate for the considered nonlinear systems.  This representation eliminates the need for explicit identification of the system matrices of the Koopman surrogate. Within this framework, the lifting function and the scheduling mapping required for the Koopman surrogate are learned from data using neural networks. Based on the resulting trajectory-based Koopman surrogate, we formulate a data-enabled MHE scheme in the lifted space.}
The proposed data-enabled MHE scheme estimates the states of the Koopman surrogate, based on which the original states of the nonlinear system are reconstructed. Additionally, we establish the stability of the proposed data-enabled MHE method.

{\color{black}
As a representative application of nonlinear state estimation, we further consider wastewater treatment processes, in which some key water-quality variables are typically not directly measurable online.  In such processes, certain contaminants, such as norfloxacin (\cite{li2025efficient}, \cite{xie2024efficient}) and p-nitrophenol (\cite{mei2019simultaneous}), if present in the effluent, can pose risks to the environment and public health and need to be effectively removed. However, the concentrations of these contaminants are typically not directly measurable and should be inferred from indirect measurements, such as UV-Vis spectroscopy and voltammetric response signals (\cite{li2025efficient}, \cite{xie2024efficient}, \cite{lu2024facile}). At the same time, the underlying biochemical reactions and transport phenomena in these processes are highly complex (\cite{zhao2023anhydrous}, \cite{liu2025opportunities}), which makes it difficult to develop accurate first-principles models for state estimation. These characteristics make wastewater treatment a well-suited application for the proposed data-enabled MHE framework, which bypasses explicit system identification and conducts state estimation directly from data. Accordingly, we apply the proposed data-enabled MHE to a simulated membrane-based biological wastewater treatment process proposed in \cite{guo2020nonlinear}.
}

\section{Preliminaries}
\subsection{Notation} 
$\mathbb{N}$ represents the set of non-negative integers. $\mathbb{N}_{[a,b]}:=\mathbb{N}\cap[a,b]$ denotes the set of non-negative integers within the interval $[a,b]$. $Q^{\dag}$ denotes the Moore-Penrose pseudoinverse of matrix $Q$. For a vector $z$, $\|z\|$ is the Euclidean norm; $\|z\|_{Q}^{2}:=z^{\top}Q z$ denotes the square of the weighted Euclidean norm; $\|z\|_{\infty}:=\max_{i}|z_{i}|$ represents the infinity norm. 
{\color{black}For a sequence $z_{k}$, $k\in\mathbb{N}_{[a,b]}$, the stacked vector $z_{[a,b]}$ is defined as $z_{[a,b]}:=[z_{a}^{\top},\ldots,z_{b}^{\top}]^{\top}$. The variable $\tilde{z}_{k}$ represents the noisy measurement of $z_{k}$.
The superscript ``o" denotes offline data, while variables without this superscript denote online data.}
Let $\mathbb{R}_{+}:=[0,\infty)$,
a function $\kappa:\mathbb{R}_{+}\rightarrow\mathbb{R}_{+}$ belongs to class $\mathcal{K}$ if it is continuous, strictly increasing, and satisfies $\kappa(0)=0$. 
The Hankel matrix of depth $N$ associated with $z_{[a,b]}$ is defined as
\begin{equation*}
    \mathcal{H}_{N}(z_{[a,b]}):=\left[\begin{array}{cccc}
        z_{a} & z_{a+1} & \ldots & z_{b-N+1} \\
        z_{a+1} & z_{a+2} & \ldots & z_{b-N+2} \\
       \vdots & \vdots & \ddots & \vdots\\
       z_{a+N-1} & z_{a+N} & \ldots & z_{b} \\
    \end{array}\right]
\end{equation*}
\subsection{System description}
We consider a class of discrete-time nonlinear systems of the following form:
    \begin{equation}\label{eq:sys}
        x_{k+1}=f(x_{k},u_{k}),\quad
        y_{k}=Cx_{k}
    \end{equation}
where $x_{k}\in\mathcal{X}\subseteq\mathbb{R}^{n_{x}}$ is the vector of system states;
$u_{k}\in\mathcal{U}\subseteq\mathbb{R}^{n_{u}}$ is the vector of known inputs; $y_{k}\in\mathbb{R}^{n_{y}}$ denotes the vector of measured outputs; $f:\mathcal{X}\times\mathcal{U}\rightarrow\mathcal{X}$ is a continuous, nonlinear function characterizing the dynamics of {\color{black}the} system \eqref{eq:sys}; $C\in\mathbb{R}^{n_{y}\times n_{x}}$ is the output matrix; $\mathcal{X}$ and $\mathcal{U}$ are compact convex sets.

The nonlinear function $f$ and output matrix $C$ of {\color{black}the} system \eqref{eq:sys} are unknown. Instead, only data trajectories of the system in \eqref{eq:sys} are available. We consider a two-stage state estimation in the data-enabled framework, which consists of an offline stage for data collection and an online stage for real-time state estimation (\cite{wolff2024robust}). Specifically, in the offline stage, trajectories of noise-free inputs 
{\color{black}$u^{o}_{[0,T-1]}$}, noisy states {\color{black}$\tilde{x}^{o}_{[0,T]}$}, and noisy measured outputs {\color{black}$\tilde{y}^{o}_{[0,T]}$} are collected. We consider that the state and output measurements are contaminated by unknown but bounded measurement noise (\cite{wolff2024robust}):
\begin{align*}
    \tilde{x}^{o}_{k}=x^{o}_{k}+\varepsilon_{x,k}^{o},\quad
    \tilde{y}^{o}_{k}=y^{o}_{k}+\varepsilon_{y,k}^{o}
\end{align*}
where $x^{o}_{k}$ and $y^{o}_{k}$ are noise-free states and measured outputs of {\color{black}the} system \eqref{eq:sys}, respectively; $\varepsilon_{x,k}^{o}$ and $\varepsilon_{y,k}^{o}$ are bounded state and output measurement noise, respectively, satisfying $\|\varepsilon_{x,k}^{o}\|\leq\bar{\varepsilon}_{x}$ and $\|\varepsilon_{y,k}^{o}\|\leq\bar{\varepsilon}_{y}$, $\forall k\in\mathbb{N}$.
In the online stage, system states are no longer measurable and need to be estimated using available input and output data. 

{\color{black}
\begin{rmk}\label{rmk:1}
    In practice, full-state information can be obtained offline using laboratory analyzers or high-fidelity sensors that are unavailable during real-time operation due to time and budget limitations. For example, in chemical and biological processes, substance concentrations can be obtained by laboratory analyzers (\cite{bourgeois2001line}, \cite{nicoletti2009computational}). However, obtaining such measurements is typically time-consuming and costly, making real-time acquisition of these variables impractical. In autonomous driving, the full state of vehicles can be measured through high-fidelity sensors offline or during testing phases (\cite{turan2021data}, \cite{wolff2024robust}). However, the high cost of these sensors can limit their deployment in series production (\cite{turan2021data}).
\end{rmk}}

\subsection{LPV Koopman surrogate of controlled systems}\label{sec:koopman}
Koopman operator theory (\cite{koopman1931hamiltonian}) was originally developed for {\color{black}the} autonomous systems, which provides a promising framework for describing the dynamics of {\color{black}the} nonlinear systems using a linear representation by lifting the original system states into a higher-dimensional space. For {\color{black}the} nonlinear controlled systems in \eqref{eq:sys}, an exact Koopman model has the following form (\cite{iacob2024koopman}):
\begin{equation}\label{eq:koopman}
z_{k+1}=Az_{k}+\phi(x_{k},u_{k})u_{k}
\end{equation}
where $z_{k}=\psi(x_{k})\in\mathbb{R}^{n_{z}}$ denotes the state vector in the lifted space; $\psi:\mathbb{R}^{n_{x}}\rightarrow\mathbb{R}^{n_{z}}$ is the state lifting function; $\phi:\mathbb{R}^{n_{x}}\times\mathbb{R}^{n_{u}}\rightarrow\mathbb{R}^{n_{z}\times n_{u}}$ describes the input matrix in the lifted space. 
According to \cite{iacob2024koopman}, by performing the factorization method, the Koopman model in \eqref{eq:koopman} can be rewritten in a linear parameter-varying (LPV) form:
\begin{equation}\label{eq:koopman_p}
    z_{k+1}=Az_{k}+\beta(p_{k})u_{k}
\end{equation}
where $p_{k}=\lambda(z_{k},u_{k})\in\mathbb{R}^{n_{p}}$ is the scheduling parameter with scheduling map $\lambda:\mathbb{R}^{n_{z}}\times\mathbb{R}^{n_{u}}\rightarrow\mathbb{R}^{n_{p}}$; {\color{black}$\beta:\mathbb{R}^{n_{p}}\rightarrow\mathbb{R}^{n_{z}\times n_{u}}$ is the input matrix satisfying $\beta\circ\lambda=\phi$ where $\circ$ denotes function composition.}

\subsection{Problem formulation}
{\color{black}
The objective of this paper is to develop a data-enabled state estimation method that provides real-time state estimates for the nonlinear systems in \eqref{eq:sys} without explicitly constructing a dynamic model.
Specifically, we consider an LPV Koopman model (\cite{iacob2024koopman}) for {\color{black}the} nonlinear system \eqref{eq:sys} in a lifted state space.
Instead of explicitly identifying this model, we construct its trajectories directly from data generated by the nonlinear system \eqref{eq:sys}. Based on these trajectories, a data-enabled convex optimization-based MHE scheme is formulated to estimate the full state of the system in real time.}

\section{Trajectory-based representation of surrogate}

In this work, we consider $\beta(p_{k})$ in \eqref{eq:koopman_p} to have affine dependence on $p_{k}$ (\cite{toth2011state}): 
\begin{equation}\label{eq:affine}
    \beta(p_{k})=B_{0}+\sum_{i=1}^{n_{p}}p_{k,i}B_{i}
\end{equation}
where $p_{k,i}$ denotes the $i$th element of $p_{k}$ and $B_{i}\in\mathbb{R}^{n_{z}\times n_{u}}$ for $i\in\mathbb{N}_{[1,n_{p}]}$. By denoting $\tilde{B}=[B_{1},\ldots,B_{n_{p}}]$, the LPV representation in \eqref{eq:koopman_p} can be reformulated as 
\begin{equation}\label{eq:lpv}
    z_{k+1}=Az_{k}+B_{0}u_{k}+\tilde{B}(p_{k}\otimes u_{k})
\end{equation}
{\color{black}where $\otimes$ denotes the Kronecker product.}
Define $B=[B_{0},\tilde{B}]$, $v_{k}:=p_{k}\otimes u_{k}$, and $\mathbf{u}_{k}=[u_{k}^{\top},v_{k}^{\top}]^{\top}$. Before proceeding further, we introduce the necessary definition and assumptions for deriving Theorem \ref{thm:fundamental}.

\begin{definition}(\cite{willems2005note})
{\color{black}$u_{[0,T-1]}$} is persistently exciting of order $N$, if $\mathrm{rank}(\mathcal{H}_{N}(u_{[0,T-1]}))=Nn_{u}$.
\end{definition}

\begin{assumption}\label{ass:controllable}
    The pair $(A,B)$ is controllable and the offline augmented input sequence {\color{black}$\mathbf{u}_{[0,T-1]}$} is persistently exciting of order $N+n_{z}+1$.
\end{assumption}

\begin{assumption}\label{ass:exact}
    There exists an exact LPV Koopman surrogate \eqref{eq:lpv} for {\color{black}the} nonlinear system \eqref{eq:sys}.
\end{assumption}

\begin{assumption}\label{ass:D}
    There exists a reconstruction matrix $D\in\mathbb{R}^{n_{x}\times n_{z}}$ 
    for \eqref{eq:lpv},
    such that:
    \begin{equation}\label{eq:D}
        x_{k}=D\psi(x_{k}),~k\in\mathbb{N}
    \end{equation}
\end{assumption}

{\color{black} Assumption \ref{ass:controllable} characterizes the richness of the collected data, which ensures that the collected data are sufficiently rich to capture the system dynamics.}
Following \cite{xiong2025data}, Assumption \ref{ass:exact} requires the existence of an exact Koopman surrogate of \eqref{eq:sys}. Based on \cite{iacob2025learning}, Assumption \ref{ass:D} requires that states $x_{k}$ of \eqref{eq:sys} can be reconstructed from states $z_{k}=\psi(x_{k})$ of Koopman surrogate \eqref{eq:lpv}. Note that Assumptions \ref{ass:exact}-\ref{ass:D} are only utilized in Theorem \ref{thm:fundamental} to establish a trajectory-based representation of \eqref{eq:lpv}; these assumptions are not required for the stability analysis in Theorem \ref{thm:stability}.

Under Assumption \ref{ass:D}, it follows from \eqref{eq:sys} and \eqref{eq:D} that
\begin{equation}\label{eq:CD}
y_{k} = CD\psi(x_{k})
\end{equation}

In the following, we consider the Koopman surrogate consisting of \eqref{eq:lpv} and \eqref{eq:CD}, which has the same inputs and outputs as the original nonlinear system \eqref{eq:sys}. 


{\color{black}To establish Theorem \ref{thm:fundamental}, we impose the assumptions stated in Assumptions \ref{ass:controllable}-\ref{ass:D}. Specifically, the augmented system \eqref{eq:lti} is assumed to be controllable, and the offline augmented input sequence is persistently exciting of sufficient order. Moreover, we assume the existence of an exact Koopman surrogate for the nonlinear system in \eqref{eq:sys}, as well as the existence of a reconstruction matrix that maps the lifted state to the original state space. For notational convenience, define $\boldsymbol{\psi}(x^{o}_{[0,T-1]})=[\psi(x^{o}_{0})^{\top},\ldots,\psi(x^{o}_{T-1})^{\top}]^{\top}$.}

\begin{thm}\label{thm:fundamental}
Consider an input/output/state trajectory, denoted by {\color{black}$u^{o}_{[0,T-1]}$, $y^{o}_{[0,T-1]}$, and $x^{o}_{[0,T-1]}$}, of {\color{black}the} nonlinear system \eqref{eq:sys} and a scheduling parameter trajectory {\color{black}$p^{o}_{[0,T-1]}$} of its corresponding LPV Koopman surrogate (i.e., \eqref{eq:lpv} and \eqref{eq:CD}). If Assumptions \ref{ass:controllable}-\ref{ass:D} hold, then the following holds:
\begin{itemize}
    \item[$\mathrm{(i)}$] 
        $\mathrm{rank}\left(\left[\begin{array}{c}
             \mathcal{H}_{1}({\color{black}\boldsymbol{\psi}(x^{o}_{[0,T-N])}})  \\
             \mathcal{H}_{N}(u^{o}_{[0,T-1]})\\
             \mathcal{H}_{N}(v^{o}_{[0,T-1]})
        \end{array}\right]\right)=n_{z}+Nn_{u}(1+ n_{p})$
    \item[$\mathrm{(ii)}$] 
    {\color{black}$u^{o}_{[0,T-1]}$, $y^{o}_{[0,T-1]}$, and $p^{o}_{[0,T-1]}$}
    form a trajectory of \eqref{eq:lpv} and \eqref{eq:CD}, if and only if there exists $\alpha\in\mathbb{R}^{T-N+1}$ such that 
    \begin{equation}\label{eq:fundamental_uy}
        \left[\begin{array}{c}
             \mathcal{H}_{N}(u^{o}_{[0,T-1]})  \\
             \mathcal{H}_{N}(v^{o}_{[0,T-1]})  \\
             \mathcal{H}_{N}(y^{o}_{[0,T-1]})
        \end{array}\right] \alpha= \left[\begin{array}{c}
             u_{[0,N-1]}  \\
             v_{[0,N-1]}  \\
             y_{[0,N-1]}
        \end{array}\right]
    \end{equation} 
\end{itemize}
\end{thm}
\begin{pf}
{\color{black}By using the definitions $v_{k}:=p_{k}\otimes u_{k}$, $B:=[B_{0},\tilde{B}]$, and $\mathbf{u}_{k}=[u_{k}^{\top},v_{k}^{\top}]^{\top}$, the Koopman surrogate \eqref{eq:lpv} can be rewritten in the form of 
\begin{align}\label{eq:lti}
    z_{k+1}&=Az_{k}+B_{0}u_{k}+\tilde{B}v_{k}\nonumber\\
    &=Az_{k}+B\mathbf{u}_{k}
\end{align}

By applying Theorem~1 of \cite{van2020willems} to \eqref{eq:lti}, it follows:
\begin{equation}\label{eq:response_uy_compact}
    \left[\begin{array}{c}
         \mathcal{H}_{N}(\mathbf{u}^{o}_{[0,T-1]})  \\
         \mathcal{H}_{N}(y^{o}_{[0,T-1]})
    \end{array}\right] \alpha= \left[\begin{array}{c}
         \mathbf{u}_{[0,N-1]}  \\
         y_{[0,N-1]}
    \end{array}\right]
\end{equation}
Since $\mathbf{u}_{k}=[u_{k}^{\top},v_{k}^{\top}]^{\top}$, the Hankel matrix $\mathcal{H}_{N}(\mathbf{u}^{o}_{[0,T-1]})$ and the vector $\mathbf{u}_{[0,N-1]}$ can be decomposed into two parts associated with $u_{k}$ and $v_{k}$, respectively. Consequently, we have
\begin{equation}\label{eq:response_y}
            \left[\begin{array}{c}
             \mathcal{H}_{N}(u^{o}_{[0,T-1]})  \\
             \mathcal{H}_{N}(v^{o}_{[0,T-1]})
        \end{array}\right] \alpha= \left[\begin{array}{c}
             u_{[0,N-1]}  \\
            v_{[0,N-1]}
        \end{array}\right]
\end{equation}
Substituting \eqref{eq:response_y} into \eqref{eq:response_uy_compact} yields \eqref{eq:fundamental_uy}. $\square$}
\end{pf}

{\color{black}
The augmented input vector $\mathbf{u}_{k}$ used in \eqref{eq:lti} consists of the original input $u_{k}$ and the auxiliary term $v_{k}$. Specifically, $v_{k}=p_{k}\otimes u_{k}$ captures 
how the scheduling parameter $p_{k}$ modulates the effect of the input $u_{k}$ on the system dynamics. By introducing this augmented input vector $\mathbf{u}_{k}$, the LPV Koopman surrogate \eqref{eq:lpv} can be reformulated as an LTI system with respect to the input $\mathbf{u}_{k}$, as shown in \eqref{eq:lti}. 
This reformulation is important because Willems' fundamental lemma (\cite{willems2005note}) applies only to LTI systems. Therefore, the lemma is applied not directly to the original LPV Koopman surrogate, but is applied to its equivalent LTI representation with respect to the augmented input $\mathbf{u}_{k}$.}

According to \cite{berberich2020trajectory}, the following relationship holds:
\begin{equation}\label{eq:fundamental_z}
  \mathcal{H}_{N}({\color{black}\boldsymbol{\psi}(x^{o}_{[0,T-1]})})\alpha={\color{black}\boldsymbol{\psi}(x_{[0,N-1]})} 
\end{equation}
{\color{black}where $\mathcal{H}_{N}(\boldsymbol{\psi}(x^{o}_{[0,T-1]}))\in\mathbb{R}^{Nn_{z}\times(T-N+1)}$, $\alpha\in\mathbb{R}^{T-N+1}$, and $\boldsymbol{\psi}(x_{[0,N-1]})\in\mathbb{R}^{Nn_{z}}$. } 
{\color{black} Based on \eqref{eq:fundamental_uy} and \eqref{eq:fundamental_z}, the Koopman surrogate \eqref{eq:lpv} can be represented as a linear combination of a single persistently excited data trajectory, which is the basis of the proposed method.}

It is worth noting that the outputs of the Koopman surrogate are not required to be identical to those of the original nonlinear system \eqref{eq:sys}. Instead, any output that is linear with respect to $z_{k}$ can be selected to construct a trajectory-based representation for \eqref{eq:lpv}.
By selecting \eqref{eq:D} as the measurement model for \eqref{eq:lpv}, it holds that 
\begin{equation}\label{eq:fundamental_x}
\mathcal{H}_{N}(x^{o}_{[0,T-1]})\alpha=x_{[0,N-1]}
\end{equation}
{\color{black}where
$\mathcal{H}_{N}(x^{o}_{[0,T-1]})\in\mathbb{R}^{Nn_{x}\times(T-N+1)}$, $\alpha\in\mathbb{R}^{T-N+1}$, and $x_{[0,N-1]}\in\mathbb{R}^{Nn_{x}}$.}

{\color{black}
\begin{rmk}\label{rmk:example}
To illustrate how the Koopman surrogate \eqref{eq:lpv} can be obtained, we use the following nonlinear system (\cite{iacob2024koopman}) as an example:
\begin{equation}
    \left[\begin{array}{c}x_{k+1, 1} \\ 
                            x_{k+1,2}\end{array}\right]
    =\left[\begin{array}{c} x_{k, 1} \\ 
                          x_{k, 2}- x_{k, 1}^2\end{array}\right]
    +\left[\begin{array}{c}1 \\ 
                            x_{k, 1}^2\end{array}\right] u_k
\end{equation}
where $x_{k,i}\in\mathbb{R}$ denotes the $i$-th element of state $x_{k}\in\mathbb{R}^{2}$, $i=1,2$; $u_{k}\in\mathbb{R}$ is the input. 
To establish the Koopman surrogate \eqref{eq:lpv}, we select the lifting functions as $z_{k}=\psi(x_{k}) = [x_{k, 1},x_{k, 2},x_{k, 1}^{2}]^{\top}$
and define the scheduling parameter as $p_{k}=\lambda(z_{k},u_{k})=[z_{k}^{\top},u_{k}]^{\top}$.
Then, the Koopman model \eqref{eq:koopman_p} has the following form:
\begin{equation*}
    z_{k+1} = \underbrace{\left[\begin{array}{ccc}
        1 & 0 & 0  \\
        0 & 1 & -1  \\ 
        0 & 0 & 1 
    \end{array}\right]}_{A}z_{k}+\underbrace{\left[\begin{array}{ccc}
        1  \\
        x_{k,1}^{2}   \\ 
        2x_{k,1} + u_{k} 
    \end{array}\right]}_{\beta(p_{k})}u_{k}
\end{equation*}
Moreover, it can be further reformulated in the form of  \eqref{eq:lpv}:
\begin{equation*}
    z_{k+1} = \underbrace{\left[\begin{array}{ccc}
        1 & 0 & 0  \\
        0 & 1 & -1  \\ 
        0 & 0 & 1 
    \end{array}\right]}_{A}z_{k}+\underbrace{\left[\begin{array}{cccc}
          1  \\ 
          0  \\
          0 
    \end{array}\right]}_{B_{0}}u_{k}+\underbrace{\left[\begin{array}{cccc}
         0 & 0 & 0 & 0 \\
         0 & 0 & 1 & 0 \\ 
         2 & 0 & 0 & 1
    \end{array}\right]}_{\tilde{B}}(p_{k}\otimes u_{k})
\end{equation*}

It is noted that the explicit forms of these matrices for {\color{black}the} nonlinear systems \eqref{eq:sys} are typically difficult to obtain.
These matrices can be identified numerically from data. Specifically, when the lifting functions and the scheduling mapping are determined, the matrices of the Koopman surrogate \eqref{eq:lpv} can be obtained by solving a least-squares problem of the form (\cite{eyubouglu2024data}, \cite{eyuboglu2025koopman}):
\begin{equation*}
    \min_{A,B_{0},\tilde{B}} \sum_{k=1}^{M-1}\|z_{k+1}-(Az_{k}+B_{0}u_{k}+\tilde{B}(p_{k}\otimes u_{k}))\|^{2}
\end{equation*}
where $M$ is the number of data points. Alternatively, inspired by \cite{shi2022deep}, \cite{li2025mamko}, and \cite{han2026mako}, when the lifting functions and scheduling mappings are approximated using neural networks, these matrices can also be treated as trainable parameters and learned jointly with the neural network parameters.
\end{rmk}
Remark \ref{rmk:example} is only intended to illustrate how the Koopman matrices in \eqref{eq:lpv} can be obtained. However, we note that the proposed data-enabled MHE approach does not require explicit identification of the Koopman matrices.
}

\section{Learning-based Data-enabled moving horizon estimation}
{\color{black}
In this section, we introduce the proposed data-enabled MHE framework, including the neural network training procedure, the formulation of the data-enabled MHE, and the stability analysis of the proposed method. 

A step-by-step illustration of the proposed data-enabled MHE is presented in Fig. \ref{Fig.Flowchart}. First, the original nonlinear system \eqref{eq:sys} is mapped to a lifted space, which results in a Koopman surrogate \eqref{eq:lpv}. In this step, neural networks are utilized to approximate the lifting functions and the scheduling mapping required to construct the Koopman surrogate \eqref{eq:lpv}.
Second, a trajectory-based representation of the Koopman surrogate \eqref{eq:lpv} is constructed directly from data, which bypasses explicit modeling. Third, based on this trajectory-based representation, the proposed data-enabled MHE is formulated in the lifted space to estimate the lifted states. Finally, the estimated lifted states are mapped back to the original space using the reconstruction matrix $D$.}

\begin{figure}[tttt]
\color{black}
  \centering
  \includegraphics[width=0.5\textwidth]{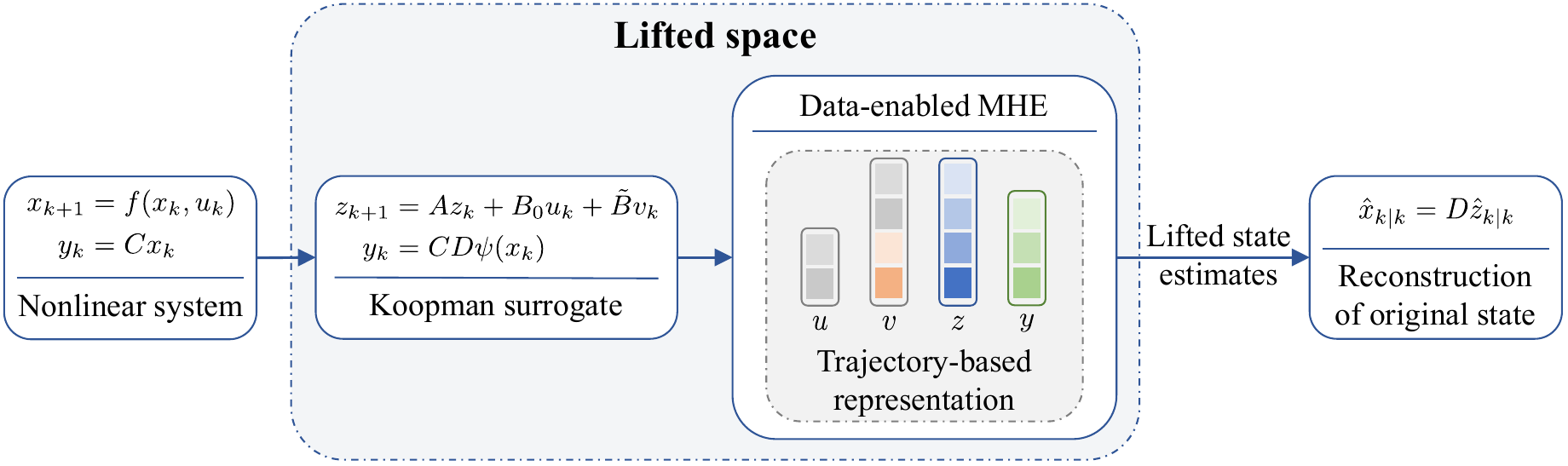}
  \caption{Flowchart of the proposed data-enabled MHE.}\label{Fig.Flowchart}
\end{figure}
\subsection{Training process}\label{sec:3.3}

According to Theorem \ref{thm:fundamental}, it is sufficient to describe \eqref{eq:lpv} using its input, output, state, and parameter trajectories without the need to identify an explicit representation. {\color{black} Since full-state information is available during the offline training stage, we choose the system state $x_{k}$ as the output of the Koopman surrogate (i.e., \eqref{eq:D} serving as the measurement model).}
{\color{black} In the Koopman surrogate \eqref{eq:lpv}-\eqref{eq:D}, the input $u_{k}$ and output $x_{k}$ are directly obtained from {\color{black}the} system \eqref{eq:sys}. In contrast, the lifted state $z_{k}$ and the scheduling parameter $p_{k}$ are not directly measurable.
}
Note that both $z_{k}$ and $p_{k}$ can be determined from $x_{k}$ and $u_{k}$, if the lifting function $\psi$ and the scheduling mapping $\lambda$ are known. {\color{black} Therefore, the objective of this section is to approximate functions $\psi$ and $\lambda$ using two dense neural networks (DNNs) (Rumelhart et al. (1986)). Based on \eqref{eq:NN_z} and \eqref{eq:NN_p}, the corresponding trajectories of $z_{k}$ and $p_{k}$ can be constructed.}

A training dataset $\mathcal{D}$ consists of 
two open-loop input and state trajectories: 1) {\color{black}$u^{o}_{[0,T-1]}$ and $x^{o}_{[0,T]}$} used to construct the Hankel matrices; 2) {\color{black}$u_{[0,N-1]}$ and $x_{[0,N]}$} used for online prediction. Fig. \ref{Fig.schematic} shows an illustrative diagram of the training process of the proposed learning-based data-enabled MHE.

{\color{black}The trained neural network $\psi_{\theta}$ projects the original state $x_{k}$ to obtain the lifted state $z_{k}$ as follows:}
\begin{equation}\label{eq:NN_z}
    z_{k}=\psi_{\theta}(x_{k}|\theta)
\end{equation}
where $\theta$ denotes the trainable parameters of neural network $\psi_{\theta}$. 
{\color{black}The trained neural network $\lambda_{\sigma}$ generates the scheduling parameter $p_{k}$ as follows: }
\begin{equation}\label{eq:NN_p}
    p_{k}=\lambda_{\sigma}(z_{k},u_{k}|\sigma)
\end{equation}
where $\sigma$ includes the trainable parameters of neural network $\lambda_{\sigma}$. Moreover, the reconstruction matrix $D$, which reconstructs the original system state $x_{k}$ from the lifted state $z_{k}$, is also treated as part of the trainable parameters.

The optimization problem associated with offline training is formulated as follows:
\begin{equation}
    \min_{\theta,\sigma,D} \mathcal{L} = \mathcal{L}_{1} + \mathcal{L}_{2}
\end{equation}
{\color{black}where $\mathcal{L}_{1}$ and $\mathcal{L}_{2}$ are defined as follows:
\begin{subequations}\label{eq:cost}
\begin{align}
\mathcal{L}_{1}&=\mathbb{E}_{\mathcal{D}}\sum_{i=0}^{N}\|Dz_{i}-x_{i}\|\\
\mathcal{L}_{2}&=\mathbb{E}_{\mathcal{D}}\sum_{i=0}^{N}\Bigg\|D\mathcal{H}_{1}(z_{[i,i+T-N]}^{o})\\
    &\quad\times\left[\begin{array}{cc}
         \mathcal{H}_{N}(u_{[0,T-1]}^{o})\\
         \mathcal{H}_{N}(v_{[0,T-1]}^{o})\\
         \mathcal{H}_{N+1}(x_{[0,T]}^{o})
    \end{array}\right]^{\dag}\left[\begin{array}{cc}
         u_{[0,N-1]}\\
         v_{[0,N-1]}\\
         x_{[0,N]}
    \end{array}\right]-x_{i}\Bigg\|
\end{align}
\end{subequations}
where $\mathbb{E}_{\mathcal{D}}$ represents the expectation taken over the dataset $\mathcal{D}$.}
{\color{black}In \eqref{eq:cost}, $\mathcal{L}_{1}$ penalizes the error between the states reconstructed by matrix $D$ and the actual states. $\mathcal{L}_{2}$ 
penalizes the discrepancy between the states predicted by the trajectory-based representation of Koopman surrogate \eqref{eq:lpv}-\eqref{eq:D} and the actual states of the nonlinear system \eqref{eq:sys}.}

\begin{figure}[tttt]
\color{black}
  \centering
  \includegraphics[width=0.46\textwidth]{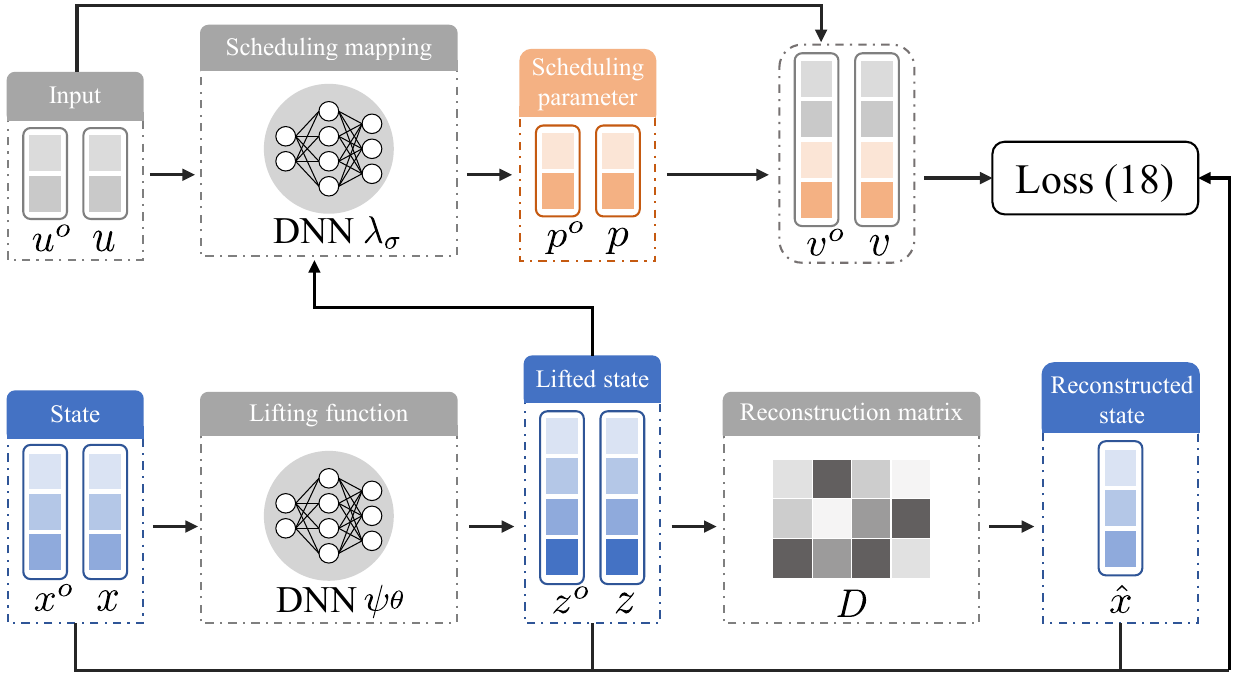}
  \caption{An illustrative diagram of the proposed data-enabled MHE.}\label{Fig.schematic}
\end{figure}


\subsection{Data-enabled MHE in lifted space}
In practice, modeling residuals inevitably arise from the LPV approximation in \eqref{eq:affine} and/or from the neural network-based approximations of the lifting function and scheduling mapping. To account for these modeling residuals, we suppose that an exact LPV Koopman surrogate can be expressed in the following form:
\begin{subequations}\label{eq:surrogate}
\begin{align}
     z_{k+1}&=Az_{k}+B_{0}u_{k}+\tilde{B}(p_{k}\otimes u_{k})+w_{k}\label{eq:surrogate_1}\\
     y_{k}&=CDz_{k}+d_{k}\label{eq:surrogate_2}\\
     x_{k}&=Dz_{k}+r_{k}\label{eq:surrogate_3}
\end{align}
\end{subequations}
where $w_{k}$, $d_{k}$, and $r_{k}$ represent unknown modeling residuals.
The lifted state, measured output, and reconstructed state of the nominal version of \eqref{eq:surrogate}, which includes \eqref{eq:lpv}, \eqref{eq:D}, and \eqref{eq:CD}, are denoted by $\check{z}_{k}$, $\check{y}_{k}$, and $\check{x}_{k}$, respectively.

\begin{assumption}\label{ass:continuous}
    The scheduling mapping $\lambda_{\sigma}$ is continuous, and lifting function $\psi_{\theta}$ is locally Lipschitz continuous on $\mathcal{X}$, which satisfies $\|\psi_{\theta}(x)-\psi_{\theta}(\check{x})\|\leq L_{\psi}\|x-\check{x}\|$, $\forall x,\check{x}\in\mathcal{X}$.
\end{assumption}

\begin{proposition}\label{prop:bounded}
    If Assumption \ref{ass:continuous} holds, there exist bounds $\bar{w}$, $\bar{d}$, $\bar{r}$, such that $\|w_{k}\|\leq\bar{w}$, $\|d_{k}\|\leq\bar{d}$, and $\|r_{k}\|\leq\bar{r}$ hold for all $k\in\mathbb{N}$.
\end{proposition}
\begin{pf}
    This proposition can be proven following the proof of Proposition~2 in \cite{zhang2022robust}.
\end{pf}

{\color{black}During the offline stage, both {\color{black}the} Koopman surrogate in \eqref{eq:surrogate} and its nominal version are initialized at the same initial condition, $z_{0}^{o}=\check{z}^{o}_{0}$. Subsequently, they evolve under the same input sequence {\color{black}$\mathbf{u}^{o}:=u^{o}_{[0,T-1]}$} and scheduling parameter sequence {\color{black}$\mathbf{p}^{o}:=p^{o}_{[0,T-1]}$}.}
{\color{black} We define the offline state deviation as $\Delta z_{k}^{o}:=z_{k}^{o}-\check{z}^{o}_{k}=\sum_{i=0}^{k-1}A^{k-1-i}w^{o}_{i}$ and the output deviation as $\Delta y^{o}_{k}:=y^{o}_{k}-\check{y}^{o}_{k}=CD\Delta z^{o}_{k}+d^{o}_{k}$. Proposition \ref{prop:bounded} implies that both deviations are bounded. Specifically, it holds that $\|\Delta z_{k}^{o}\|\leq\bar{\Delta}_{z}^{o}$ and $\|\Delta y_{k}^{o}\|\leq\bar{\Delta}_{y}^{o}$, $\forall k\in\mathbb{N}_{[0,T]}$, where $\bar{\Delta}^{o}_{z}:=\sum_{i=0}^{T-1}\|A\|^{k-1-i}\bar{w}$ and $\bar{\Delta}^{o}_{y}=\|C\|\|D\|\bar{\Delta}_{z}^{o}+\bar{d}$.}

{\color{black}
For online implementation, only the measured output $y_{k}$ is available. Accordingly, we select $y_{k}$ as the output of the Koopman surrogate (i.e., \eqref{eq:CD} as the measurement model).}
{\color{black}
According to Theorem \ref{thm:fundamental}, the input, output, state, and parameter trajectories of the Koopman surrogate in \eqref{eq:lpv} and \eqref{eq:CD} provide a trajectory-based representation. Based on this representation, we formulate a data-enabled MHE in the lifted space.} Specifically, at each time instant $k\geq N$, the estimator provides state estimates by solving the following optimization problem:
\begin{subequations}\label{eq:mhe}
\begin{align}
&\quad\quad\min_{\hat{z}_{[k-N,k]|k},\pi^{y}_{[k-N,k]|k},\pi^{z}_{[k-N,k]|k},\alpha_{k}} V_{k} \nonumber\\
\mathrm{s.t.}\,&
        \left[\begin{array}{c}
             \mathcal{H}_{N}(u_{[0,T-1]}^{o})  \\
             \mathcal{H}_{N}(v_{[0,T-1]}^{o})  \\
             \mathcal{H}_{N+1}(\tilde{y}_{[0,T]}^{o})\\
             \mathcal{H}_{N+1}(\tilde{z}_{[0,T]}^{o})
        \end{array}\right] \!\!\alpha_k\!\!= \!\left[\begin{array}{c}
             u_{[k-N,k-1]}  \\
             v_{[k-N,k-1]}\\
             \tilde{y}_{[k-N,k]}\!-\!\pi^{y}_{[k-N,k]|k}\\
             \hat{z}_{[k-N,k]|k}\!+\!\pi^{z}_{[k-N,k]|k}
        \end{array}\right]\label{eq:mhe_1}\\
        &\,D\hat{z}_{j|k}\in\mathcal{X},~j\in\mathbb{N}_{[k-N,k]}\label{eq:mhe_2}
\end{align}
\text{with}
\begin{align}\label{eq:mhe_obj}
    V_{k}&=\lambda_{z}\|\hat{z}_{k-N|k}-\bar{z}_{k-N}\|_{P}^{2}+\sum_{j=k-N}^{k}\|\pi^{z}_{j|k}\|^{2}_{Q}\\
&\quad+\sum_{j=k-N}^{k}\|\pi^{y}_{j|k}\|^{2}_{R}+\lambda_{\alpha}(\bar{\varepsilon}_{x}+\bar{\varepsilon}_{y}+\bar{\Delta}_{z}^{o}+\bar{\Delta}_{y}^{o})\|\alpha_{k}\|^{2}\nonumber
\end{align}
\end{subequations}
where $\hat{z}_{j|k}$ denotes the state estimate of \eqref{eq:lpv} for time instant $j$, computed at time instant $k$; $\pi^{y}_{[k-N,k]|k}$ is a sequence of fitting errors that account for the effect of offline/online measurement noise on measured outputs and the modeling residual $d_{[k-N,k]}$; $\pi^{z}_{[k-N,k]|k}$ is a sequence of slack variables that account for the noise in the state trajectory $\tilde{z}_{[0,T]}^{o}$ and the modeling residual $w_{[k-N,k]}$; $u_{[k-N,k-1]}$ contains the online noise-free known inputs from time instant $k-N$ to $k-1$; $p_{j}=\lambda_{\sigma}(\hat{z}_{j|j}, u_{j})$ is the approximate scheduling parameter at time instant $j$; $v_{j}=p_{j}\otimes u_{j}$; $\tilde{y}_{[k-N,k]}$ contains the noisy online measured outputs from time instant $k-N$ to $k$;  $\bar{z}_{k-N}=\hat{z}_{k-N|k-N}$ is the prior estimate of $z_{k-N}$. {\color{black} To ensure the convexity of the MHE problem \eqref{eq:mhe}, the weighting matrices $P$, $Q$, and $R$ are required to be positive definite, and the scalars $\lambda_{z}$, $\lambda_{\alpha}$ satisfy $\lambda_{z}$, $\lambda_{\alpha}>0$.}
Note that when $k<N$, the data-enabled MHE in \eqref{eq:mhe} reduces to a full-information estimator, obtained by replacing $N$ with $k$ in \eqref{eq:mhe}. 
Based on \eqref{eq:mhe}, the state estimate $\hat{x}_{k|k}$ of {\color{black}the} system \eqref{eq:sys} is obtained as $\hat{x}_{k|k} = D\hat{z}_{k|k}$.

{\color{black}
\begin{rmk}
The finite size of datasets may affect the validity of the proposed data-enabled MHE method, as this depends on whether the collected offline augmented input sequence satisfies the required persistent excitation condition. As discussed in Remark~1 of \cite{zhang2023dimension}, the offline dataset size $T$ must satisfy $T\geq n_{z}+Nn_{u}(1+n_{p})+N-1$. In addition, the offline augmented input sequence must be persistently exciting of order $N+n_{z}+1$. 
Although a larger dataset makes these excitation conditions more likely to be satisfied, it also increases the dimension of the optimization variable $\alpha_{k}$, which will increase the computational complexity and potentially limit practical deployment in resource-limited settings.
\end{rmk}}

\subsection{Stability analysis}
Let $z(k;z_{0},\mathbf{w})$ denote the solution to \eqref{eq:surrogate_1} at time instant $k$, given the initial state $z_{0}$, the input sequence $\mathbf{u}:=\{u_{0},u_{1},\ldots\}$, the scheduling parameter sequence $\mathbf{p}:=\{p_{0},p_{1},\ldots\}$, and modeling residual sequence $\mathbf{w}:=\{w_{0},w_{1},\ldots\}$. We omit $\mathbf{u}$ and $\mathbf{p}$ from the notation $z(k;z_{0},\mathbf{w})$ for simplicity. Affected by the modeling residual $\mathbf{d}=\{d_{0},d_{1},\ldots\}$, the corresponding measured output of \eqref{eq:surrogate_2} at time instant $k$ is given by $h(z(k;z_{0},\mathbf{w}),\mathbf{d})$. 
\begin{definition}(\textbf{i-UEIOSS}) (\cite{alessandri2025robust})\label{def:ioss}
    System \eqref{eq:surrogate_1}-\eqref{eq:surrogate_2} is i-UEIOSS if there exist $c_{z},c_{w},c_{h}>0$ and $\xi\in (0,1)$ such that, for any two solutions $z(k;z_{0},\mathbf{w})$ and $z(k;\check{z}_{0},\check{\mathbf{w}})$, it holds for all $k\in\mathbb{N}$ that:
    \begin{align}\label{eq:e-uoss}
        &\quad\|z(k;z_{0},\mathbf{w})-z(k;\check{z}_{0},\check{\mathbf{w}})\|\nonumber\\
        &\leq c_{z}\|z_{0}-\check{z}_{0}\|\xi^{k}+c_{w}\sum_{\tau=0}^{k-1}\|w_{\tau}-\check{w}_{\tau}\|\xi^{k-1-\tau}\\
        &\quad+c_{h}\sum_{\tau=0}^{k-1}\|h(z(\tau;z_{0},\mathbf{w}),\mathbf{d})-h(z(\tau;\check{z}_{0},\check{\mathbf{w}}),\check{\mathbf{d}})\|\xi^{k-1-\tau}\nonumber
    \end{align}
\end{definition}
Note that the i-UEIOSS in \cite{alessandri2025robust} is formulated in terms of squared norms.


\begin{assumption}\label{ass:detectable}
    The pair $(A,CD)$ is detectable.
\end{assumption}
Assumption \ref{ass:detectable} indicates the incremental uniform exponential input/output-to-state stability (i-UEIOSS) of {\color{black}the} system \eqref{eq:surrogate_1}-\eqref{eq:surrogate_2}. 


{\color{black}\begin{assumption}\label{ass:parameter}
The estimation horizon \(N\), the weighting matrices \(P\), \(Q\), \(R\), and the tuning parameters \(\lambda_z\), \(\lambda_{\alpha}\) satisfy
\begin{align*}
&c_z \xi^{N} \leq \lambda_z \underline{p},~c_z + 1 \leq \underline{q},~c_h \leq \underline{r}, \\
&\max \big\{\max\{L_{\psi},1\}(c_z + 1)\sqrt{n_z (T - N + 1)}, \\
&\quad \quad\quad c_h N \sqrt{n_y (T - N + 1)}
\big\}\leq \lambda_{\alpha}
\end{align*}
where $\underline{p}$, $\underline{q}$, and $\underline{r}$ are the smallest eigenvalues of matrices $P$, $Q$, and $R$, respectively.
\end{assumption}}

{\color{black}Before deriving Theorem \ref{thm:stability}, we impose the assumptions stated in Assumption \ref{ass:controllable} and Assumptions \ref{ass:continuous}-\ref{ass:parameter}. Specifically, the augmented system \eqref{eq:lti} is assumed to be controllable, and the offline augmented input sequence is persistently exciting of sufficient order. The scheduling mapping $\lambda_{\sigma}$ is assumed to be continuous, and the lifting function $\psi_{\theta}$ is locally Lipschitz continuous. Additionally, the pair $(A, CD)$ is assumed to be detectable. Finally, the parameters of the MHE design \eqref{eq:mhe} are assumed to be appropriately tuned.}
\begin{thm}\label{thm:stability}
Consider {\color{black}the} system \eqref{eq:sys} subject to offline measurement noise (i.e., $\varepsilon_{x,k}^{o}$ and $\varepsilon_{y,k}^{o}$), online measurement noise $\varepsilon_{y,k}$, and modeling residuals (i.e., $w_{k}$, $d_{k}$, $r_{k}$) with respect to its corresponding Koopman surrogate \eqref{eq:surrogate}. If Assumption \ref{ass:controllable} and Assumptions \ref{ass:continuous}-\ref{ass:parameter} hold, then there exist $\gamma_{z}$, $\gamma_{w}$, $\gamma_{\varepsilon}$, $\gamma_{d}>0$, $\theta_{z}\in(0,1)$, and $\theta\in\mathcal{K}$, such that the following condition holds for all $k\in\mathbb{N}$:
\begin{align}\label{eq:thm_pres}
    &\|x_{k}-\hat{x}_{k|k}^{*}\|\leq
    \gamma_{z}\|x_{0}-\bar{x}_{0}\|\theta_{z}^{k}+\theta(\delta)\\
    &+\gamma_{w}\sum_{\tau=0}^{k}\|w_{k-\tau}\|\theta_{z}^{\tau}+\gamma_{\varepsilon}\sum_{\tau=0}^{k}\|\varepsilon_{y,k-\tau}\|\theta_{z}^{\tau}+\gamma_{d}\sum_{\tau=0}^{k}\|d_{k-\tau}\|\theta_{z}^{\tau}\nonumber
\end{align}
where $\delta=\max\{\bar{\varepsilon}_{x},\bar{\varepsilon}_{y},\bar{\Delta}_{z}^{o},\bar{\Delta}_{y}^{o},\bar{r}\}$.
\end{thm}
\begin{pf}
This proof is based on the proof of Theorem~2 in \cite{wolff2024robust}. 
For notational simplicity, we denote $\bar{\lambda}_{\alpha}=\lambda_{\alpha}(\bar{\varepsilon}_{x}+\bar{\varepsilon}_{y}+\bar{\Delta}_{z}^{o}+\bar{\Delta}_{y}^{o})$.

\textbf{(I) Lower and upper bounds for optimal cost $V_{k}^{*}$}

A lower bound for optimal cost $V_{k}^{*}$ has the following form:
\begin{align}\label{eq:lower_bound_sqrt}
   &(\eta_{N} V^{*}_{k})^{\frac{1}{2}}\geq \lambda_{z}\underline{p}\|\hat{z}_{k-N|k}^{*}-\bar{z}_{k-N}\|+\bar{\lambda}_{\alpha}\|\alpha_{k}^{*}\|\\
   &\quad\quad\quad+\underline{q}\sum_{j=k-N}^{k}\|\pi^{z,*}_{j|k}\|+\sum_{j=k-N}^{k}\underline{r}\|\pi^{y,*}_{j|k}\|\nonumber
\end{align}
where $\eta_k = \lambda_{z}\underline{p}+\underline{q}(k+1)+\underline{r}(k+1)+\bar{\lambda}_{\alpha}$. 
Then, we leverage the trajectory of {\color{black}the} nominal system to construct a feasible solution for the proposed data-enabled MHE \eqref{eq:mhe}, which provides an upper bound for optimal cost $V_{k}^{*}$.
{\color{black}Based on Assumption \ref{ass:controllable} and the condition $\mathrm{(i)}$ in Theorem \ref{thm:fundamental}, we select $\alpha_k$ as}
\begin{equation}\label{eq:select_alpha}
    \alpha_{k}= \left[\begin{array}{c}
        \mathcal{H}_{N}(u^{o}_{[0,T-1]}) \\
        \mathcal{H}_{N}(v^{o}_{[0,T-1]}) \\
        \mathcal{H}_{1}(\check{z}^{o}_{[0,T-N]})
    \end{array}\right]^{\dag}\left[\begin{array}{c}
        u_{[k-N,k-1]} \\
        v_{[k-N,k-1]} \\
        \check{z}_{k-N}
    \end{array}\right]
\end{equation}
$\alpha_k$ satisfies:
\begin{equation*}
    \left[\begin{array}{c}
         \mathcal{H}_{N}(u^{o}_{[0,T-1]}) \\
        \mathcal{H}_{N}(v^{o}_{[0,T-1]}) \\
         \mathcal{H}_{N+1}(\check{y}^{o}_{[0,T]})\\
         \mathcal{H}_{N+1}(\check{z}^{o}_{[0,T]})
    \end{array}\right] \alpha_k= \left[\begin{array}{c}
         u_{[k-N,k-1]} \\
        v_{[k-N,k-1]} \\
         \tilde{\check{y}}_{[k-N,k]}-\varepsilon_{y,[k-N,k]}\\
         \check{z}_{[k-N,k]}
    \end{array}\right]
\end{equation*}
where $\varepsilon_{y,[k-N,k]}$ is the actual online measurement noise from time instant $k-N$ to $k$.
To satisfy \eqref{eq:mhe_1}, we select 
\begin{subequations}\label{eq:select_wv}
    \begin{align}
    \pi^{z}_{[k-N,k]}&
    =\mathcal{H}_{N+1}(\Delta z^{o}_{[0,T]})\alpha_{k}+\mathcal{H}_{N+1}( \varepsilon^{o}_{z,[0,T]})\alpha_{k}\label{eq:select_pi_z}\\
     \pi^{y}_{[k-N,k]}
     &=\varepsilon_{y,[k-N,k]}+\Delta y_{[k-N,k]}-\mathcal{H}_{N+1}(\Delta y_{[0,T]}^{o})\alpha_{k}\nonumber\\
     &\quad-\mathcal{H}_{N+1}(\varepsilon_{y,[0,T]}^{o})\alpha_{k}
     \label{eq:select_pi_y}
\end{align}
\end{subequations}
where $\varepsilon_{z,j}^{o}=\psi_{\theta}(x_{j}^{o}+\varepsilon_{x,j}^{o})-\psi_{\theta}(x_{j}^{o})$, $j\in\mathbb{N}_{[0,T]}$.
An upper bound for $V_{k}^{*}$ is established as follows:
\begin{align}\label{eq:upper_bound_sqrt}
    &(\eta_{N} V_{k}^{*})^{\frac{1}{2}}\leq\sqrt{\eta_{N} \lambda_{z}\bar{p}}\|z_{k-N}-\bar{z}_{k-N}\|+\sqrt{\eta_{N} \bar{\lambda}_{\alpha}}\|\alpha_{k}\|\nonumber\\
    &\quad\quad\quad+\sqrt{\eta_{N} \bar{q}}\sum_{j=k-N}^{k}\|\pi^{z}_{j}\|
    +\sqrt{\eta_{N}\bar{r}}\sum_{j=k-N}^{k}\|\pi^{y}_{j}\|
\end{align}
\normalsize
where $\bar{p}$, $\bar{q}$, and $\bar{r}$ are largest eigenvalues of matrices $P$, $Q$, and $R$, respectively. When $k<N$, the lower and upper bounds for the optimal cost $V_{k}^{*}$ can be derived by replacing $N$ with $k$ in \eqref{eq:lower_bound_sqrt} and \eqref{eq:upper_bound_sqrt}.

\textbf{(II) Boundedness of estimation error $\|z_{k}-\hat{z}_{k|k}^{*}\|$}

To establish an upper bound on the estimation error, we consider two trajectories of \eqref{eq:surrogate}: 1) the nominal trajectory generated using $\alpha_{k}$ as defined in \eqref{eq:select_alpha}; 2) the estimated trajectory obtained using $\alpha_{k}^{*}$. Specifically, given the initial condition $z_{k-N}$, the input sequence $\mathbf{u}$, and the scheduling parameter $\mathbf{p}$, the nominal state at time instant $k$ is denoted as $z_{k}=z(N;z_{k-N},\mathbf{0})$. The estimated state trajectory is initialized from $\check{z}_{k-N}$, evolves under the same $\mathbf{u}$ and $\mathbf{p}$, but is affected by the modeling residuals $\mathbf{w}$ and $\mathbf{d}$. The estimated state at time instant $k$ is denoted as $\check{z}_{k}=z(N;\check{z}_{k-N},\mathbf{w})$. We select $\check{z}_{k-N}$ as follows:
\begin{align}\label{eq:x2}
    \check{z}_{k-N}&=\hat{z}_{k-N|k}^{*}+\pi_{k-N|k}^{z,*}-\mathcal{H}_{1}(\Delta z^{o}_{[0,T-N]})\alpha_{k}^{*}\nonumber\\
    &\quad-\mathcal{H}_{1}(\varepsilon_{z,[0,T-N]}^{o})\alpha_{k}^{*}
\end{align}
Then, the state $\check{z}_{k}$ can be obtained as follows:
\begin{equation*}
   \check{z}_{k}=\hat{z}_{k|k}^{*}+\pi_{k|k}^{z,*}-\mathcal{H}_{1}(\Delta z^{o}_{[N,T]})\alpha_{k}^{*}-\mathcal{H}_{1}(\varepsilon_{z,[N,T]}^{o})\alpha_{k}^{*}+\Delta z_{k}
\end{equation*}
It follows that
\begin{align}\label{eq:delta xk_1}
    \|z_{k}-\hat{z}_{k|k}^{*}\|&\leq\|z_{k}-\check{z}_{k}\|+\|\pi_{k|k}^{z,*}\|+\|\mathcal{H}_{1}(\Delta z^{o}_{[N,T]})\alpha_{k}^{*}\|\nonumber\\
    &\quad+\|\mathcal{H}_{1}(\varepsilon_{z,[N,T]}^{o})\alpha_{k}^{*}\|+\|\Delta z_{k}\|
\end{align}
{\color{black}According to Assumption \ref{ass:detectable} and \eqref{eq:x2}, it holds that}
\begin{align}\label{eq:delta xk_2_temp}
   &\|z_{k}-\check{z}_{k}\|\leq c_{z}\|z_{k-N}-\hat{z}_{k-N|k}^{*}\|\xi^{N}+c_{z}\|\pi_{k-N|k}^{z,*}\|\xi^{N}\nonumber\\
    &\quad+c_{z}\|\mathcal{H}_{1}(\Delta z^{o}_{[0,T-N]})\alpha_{k}^{*}\|\xi^{N}+c_{z}\|\mathcal{H}_{1}(\varepsilon_{z,[0,T-N]}^{o})\alpha_{k}^{*}\|\xi^{N}\nonumber\\
    &\quad+c_{h}\sum_{\tau=0}^{N-1}\|h(z(\tau;z_{k-N},\mathbf{0}),\mathbf{0})-h(z(\tau;\check{z}_{k-N},\mathbf{w}),\mathbf{d})\|\nonumber \\
    &\quad\times\xi^{N-\tau-1}+c_{w}\sum_{\tau=0}^{N-1}\|w_{k-N+\tau}\|\xi^{N-\tau-1}
\end{align}

Then, we establish upper bounds for each term on the right-hand side of \eqref{eq:delta xk_2_temp}. It follows from $0<\xi<1$ that $c_{z}\|\pi_{k-N|k}^{z,*}\|\xi^{N}+\|\pi_{k|k}^{z,*}\|\leq (c_{z}+1)\|\pi_{[k-N,k]|k}^{z,*}\|$.
Note that $h(z(\tau;z_{k-N},\mathbf{0}),\mathbf{0})=\tilde{y}_{k-N+\tau}-\Delta y_{k-N+\tau}-\varepsilon_{y,k-N+\tau}$ and $h(z(\tau;\check{z}_{k-N},\mathbf{w}),\mathbf{d})=\mathcal{H}_{1}(\check{y}^{o}_{[\tau,T-N+\tau]})\alpha_{k}^{*}+\Delta y_{k-N+\tau}=\tilde{y}_{k-N+\tau}-\pi^{y,*}_{k-N+\tau|k}-\mathcal{H}_{1}(\varepsilon^{o}_{y,[\tau,T-N+\tau]})\alpha^{*}_{k}-\mathcal{H}_{1}(\Delta y^{o}_{[\tau,T-N+\tau]})\alpha^{*}_{k}+\Delta y_{k-N+\tau}$
hold for $\tau\in\mathbb{N}_{[0,N]}$ with $\Delta y_{k-N+\tau}=CD\Delta z_{k-N+\tau}+d_{k-N+\tau}$.
One can obtain that
\begin{align}\label{eq:h_term}
    &\quad \|h(z(\tau;z_{k-N},\mathbf{0}),\mathbf{0})-h(z(\tau;\check{z}_{k-N},\mathbf{w}),\mathbf{d})\|\nonumber\\
    &\leq2\|\Delta y_{k-N+\tau}\|+\|\varepsilon_{y,k-N+\tau}\|+\|\pi^{y,*}_{k-N+\tau|k}\|\\
    &\quad+\|\mathcal{H}_{1}(\varepsilon^{o}_{y,[\tau,T-N+\tau]})\alpha^{*}_{k}\|+\|\mathcal{H}_{1}(\Delta y^{o}_{[\tau,T-N+\tau]})\alpha^{*}_{k}\|\nonumber
\end{align}
Finally, we establish upper bounds for the terms associated with $\alpha_{k}^{*}$.
Following \cite{wolff2024robust}, we have
\begin{equation}\label{eq:bound_epsilon_y}
    \|\mathcal{H}_{1}(\varepsilon_{y,[\tau,T-N+\tau]}^{o})\|\leq\sqrt{n_{y}(T-N+1)}\bar{\varepsilon}_{y}
\end{equation}
We can obtain upper bounds for $\|\mathcal{H}_{1}(\Delta y_{[\tau,T-N+\tau]}^{o})\|$ and $\|\mathcal{H}_{1}(\Delta z_{[\tau,T-N+\tau]}^{o})\|$ in a similar way.
{\color{black}By leveraging Lipschitz continuity of the lifting function $\psi_{\theta}$ from Assumption \ref{ass:continuous}, it follows that
\begin{equation}
    \|\varepsilon_{z,j}^{o}\|=\|\psi_{\theta}(x_{j}^{o}+\varepsilon_{x,j}^{o})-\psi_{\theta}(x_{j}^{o})\|\leq L_{\psi}\bar{\varepsilon}_{x}
\end{equation}
} Then, we have that
\begin{equation}\label{eq:bound_epsilon_x_1}
    \|\mathcal{H}_{1}(\varepsilon_{z,[N,T]}^{o})\|\leq\sqrt{n_{z}(T-N+1)}L_{\psi}\bar{\varepsilon}_{x}
\end{equation}
Selecting $c_{\alpha}=\max\big\{(c_{z}+1)\sqrt{n_{z}(T-N+1)}L_{\psi}$, $(c_{z}+1)\sqrt{n_{z}(T-N+1)},c_{h}N\sqrt{n_{y}(T-N+1)}\big\}$ and 
substituting \eqref{eq:delta xk_2_temp}-\eqref{eq:bound_epsilon_x_1} into \eqref{eq:delta xk_1} yield
\begin{align}\label{eq:delta xk_2}
    &\|z_{k}-\hat{z}_{k|k}^{*}\|\leq c_{z}\|z_{k-N}-\bar{z}_{k-N}\|\xi^{N}\!+\!c_{z}\|\hat{z}_{k-N|k}^{*}-\bar{z}_{k-N}\|\xi^{N}\nonumber\\
    &\quad\quad\quad+(c_{z}+1)\|\pi_{[k-N,k]|k}^{z,*}\|+\|\Delta z_{k}\|+\bar{c}_{\alpha}\|\alpha_{k}^{*}\|\\
    &\quad\quad\quad+c_{h}\sum_{\tau=0}^{N-1}\|\pi^{y,*}_{k-N+\tau|k}\|+c_{h}\sum_{\tau=0}^{N-1}\big(\|\varepsilon_{y,k-N+\tau}\|\nonumber\\
    &\quad\quad\quad+2\Delta y_{k-N+\tau}\|\big)\xi^{N-\tau-1}+c_{w}\sum_{\tau=0}^{N-1}\|w_{k-N+\tau}\|\xi^{N-\tau-1}\nonumber
\end{align}
where $\bar{c}_{\alpha}=c_{\alpha}(\bar{\varepsilon}_{x}+\bar{\varepsilon}_{y}+\bar{\Delta}_{y}^{o}+\bar{\Delta}_{z}^{o})$.
When $k<N$, \eqref{eq:delta xk_2} holds by replacing $N$ with $k$. 
Denote $k=\tilde{k}+jN$ with $\tilde{k}\in\mathbb{N}_{[0,N-1]}$ and $j\in\mathbb{N}$.

\textbf{(II.a) Case I: $k\geq N$}.

Select $\epsilon\in(0,1)$ and an estimation horizon $N$ such that $c_{z}\xi^{N}\leq \epsilon$. The smallest value of $N$ satisfying this condition is denoted by $N_{0}$. Select $\lambda_{\alpha}$, $P$, $Q$, and $R$ such that $\epsilon\leq\lambda_{z}\underline{p}$, $c_{z}+1\leq\underline{q}$, $c_{h}\leq\underline{r}$, and $c_{\alpha}\leq\lambda_{\alpha}$.
{\color{black} Based on Assumption \ref{ass:parameter}, \eqref{eq:lower_bound_sqrt} and \eqref{eq:upper_bound_sqrt}, we derive from \eqref{eq:delta xk_2} that}
\begin{align}\label{eq:delta xk-N+1_3}
    &\quad\|z_{\tilde{k}+N}-\hat{z}_{\tilde{k}+N|\tilde{k}+N}^{*}\|\nonumber\\
    &\leq
    (\epsilon+\sqrt{\eta_{N}\lambda_{z}\bar{p}})\|z_{\tilde{k}}-\bar{z}_{\tilde{k}}\|+\sqrt{\eta_{N}\bar{\lambda}_{\alpha}}\|\alpha_{\tilde{k}+N}\|\\
&\quad+\sqrt{\eta_{N}\bar{q}}\sum_{\tau=0}^{N}\|\pi_{\tilde{k}+\tau}^{z}\|+\sqrt{\eta_{N}\bar{r}}\sum_{\tau=0}^{N}\|\pi^{y}_{\tilde{k}+\tau}\|+\|\Delta z_{\tilde{k}+N}\|\nonumber\\
    &\quad+\sum_{\tau=0}^{N-1}\big(c_{h}(2\|\Delta y_{\tilde{k}+\tau}\|+\|\varepsilon_{y,\tilde{k}+\tau}\|)+c_{w}\|w_{\tilde{k}+\tau}\|\big)\xi^{N-\tau-1}\nonumber
\end{align}
From \eqref{eq:select_wv} and the derivation of \eqref{eq:bound_epsilon_y}, it holds that
\begin{align}\label{eq:bound_wv}
    \|\pi^{z}_{\tilde{k}+\tau}\|
    &\leq\sqrt{n_{z}(T-N+1)}(L_{\psi}\bar{\varepsilon}_{x}+\bar{\Delta}_{z}^{o})\|\alpha_{\tilde{k}+N}\|\nonumber\\
\|\pi^{y}_{\tilde{k}+\tau}\|
     &\leq\|\varepsilon_{y,\tilde{k}+\tau}\|+\|\Delta y_{\tilde{k}+\tau}\|\\
     &\quad+\sqrt{n_{y}(T-N+1)}(\bar{\varepsilon}_{y}+\bar{\Delta}_{y}^{o})\|\alpha_{\tilde{k}+N}\|\nonumber
\end{align}
Moreover, we have $\|\Delta z_{\tilde{k}+\tau}\|\leq \omega_{A}\sum_{i=0}^{\tau-1}\|w_{\tilde{k}+i}\|$ and $\|\Delta y_{\tilde{k}+\tau}\|\leq\|C\|\|D\|\|\Delta z_{\tilde{k}+\tau}\|+\|d_{\tilde{k}+\tau}\|$ with $\omega_{A}=\max\{1,\|A\|^{N-1}\}$.
Given that $u_k\in\mathcal{U}$ and $x_k\in\mathcal{X}$, the continuity of $\psi_{\theta}$ and $\lambda_{\sigma}$ implies that $z_{k}=\psi_{\theta}(x_{k})$ and $p_{k}=\lambda_{\sigma}(z_{k},u_{k})$ are contained in compact sets. Therefore, there exist $\bar{u}$, $\bar{z}$, and $\bar{p}$, such that $\|u_{k}\|\leq\bar{u}$, $\|z_{k}\|\leq\bar{z}$, and $\|p_{k}\|\leq\bar{p}$ hold for all $k\in\mathbb{N}$. It follows that $\|v_{k}\|\leq\bar{u}\bar{p}$.
Following \cite{wolff2024robust}, there exists an upper bound $\bar{\alpha}$ such that $\|\alpha_{k}\|\leq\bar{\alpha}$.
By \eqref{eq:bound_wv}, \eqref{eq:delta xk-N+1_3} can be reformulated as follows:
\begin{align}\label{eq:delta xk-N+1_4}
    &\|z_{\tilde{k}+N}-\hat{z}_{\tilde{k}+N|\tilde{k}+N}^{*}\|\leq \omega_{z}\|z_{\tilde{k}}-\bar{z}_{\tilde{k}}\|+\omega(\varpi)\\
    &\quad+\sum_{\tau=0}^{N}(\omega_{w}\|w_{\tilde{k}+\tau}\|+\omega_{\varepsilon}\|\varepsilon_{y,\tilde{k}+\tau}\|+\omega_{d}\|d_{\tilde{k}+\tau}\|)\xi^{N-\tau-1}\nonumber
\end{align}
where $\omega_{z}=\epsilon+\sqrt{\eta_{N}\lambda_{z}\bar{p}}$, $\omega_{\varepsilon}=\sqrt{\eta_{N}\bar{r}}\xi^{1-N}+c_{h}$, $\omega_{d}=\sqrt{\eta_{N}\bar{q}}\xi^{1-N}+2c_{h}$, $\omega_{w}=c_{w}+\omega_{d}\|C\|\|D\|\omega_{A}(N+1)+\omega_{A}\xi^{1-N}$, $\varpi=\max\{\bar{\varepsilon}_{x},\bar{\varepsilon}_{y}, \bar{\Delta}_{z}^{o},\bar{\Delta}_{y}^{o}\}$, and $\omega(\varpi)=\big((\frac{2c_{\alpha}}{(c_{z}+1)}\sqrt{\eta_{N}\bar{q}}+\frac{2c_{\alpha}(N+1)}{c_{h}N}\sqrt{\eta_{N}\bar{r}})\varpi+\sqrt{\eta_{N}\bar{\lambda}_{\alpha}}\big)\bar{\alpha}$.


Since $\omega_{z}\leq\lambda_{z}\underline{p}+\sqrt{\eta_{N}\lambda_{z}\bar{p}}$, the condition $\omega_{z}<1$ can be satisfied by tuning $\lambda_{z}$.
For $j>1$, by jointly considering the fact that $\bar{z}_{\tau}=\hat{z}_{\tau|\tau}^{*}$ and \eqref{eq:delta xk-N+1_4}, it is further obtained that
\begin{align}\label{eq:delta xk-j(N+1)}
    &\|z_{\tilde{k}+jN}-\hat{z}_{\tilde{k}+jN|\tilde{k}+jN}^{*}\|\leq\|z_{\tilde{k}}-\bar{z}_{\tilde{k}}\|\omega_{z}^{j}+\frac{1}{1-\omega_{z}}\omega(\varpi)\nonumber\\
    &\quad\quad+\sum_{i=0}^{j-1}\omega_{z}^{i}\sum_{\tau=0}^{N}\big(\omega_{w}\|w_{\tilde{k}+(j-i)N-\tau}\|+\omega_{\varepsilon}\|\varepsilon_{y,\tilde{k}+(j-i)N-\tau}\|\nonumber\\
    &\quad\quad+\omega_{d}\|d_{\tilde{k}+(j-i)N-\tau}\|\big)\xi^{\tau-1}
\end{align}

\textbf{(II.b) Case II: $k< N$}.

According to the lower and upper bounds established in \cite{wolff2024robust}, one has
\begin{align*}
    &\|\hat{z}_{0|k}^{*}-\bar{z}_{0}\|\!\leq\!\sigma\|z_{0}-\bar{z}_{0}\|\!+\!\sum_{\tau=0}^{k}(\sigma_{z}\|\pi^{z}_{\tau}\|\!+\!\sigma_{y}\|\pi^{y}_{\tau}\|)\!+\!\sigma_{\alpha}\|\alpha_{k}\|
\end{align*}
where $\sigma=\sqrt{\bar{p}/\underline{p}}$, $\sigma_{z}=\sqrt{\bar{q}/\lambda_{z}\underline{p}}$, $\sigma_{y}=\sqrt{\bar{r}/\lambda_{z}\underline{p}}$, and $\sigma_{\alpha}=\sqrt{\bar{\lambda}_{\alpha}/\lambda_{z}\underline{p}}$.
Since $c_{z}\xi^{N}\leq\epsilon\leq\lambda_z\underline{p}$, applying the same procedure to derive \eqref{eq:delta xk-N+1_4} yields
\begin{align}\label{eq:38}
    \|z_{k}-\hat{z}_{k|k}^{*}\|&\leq \rho_{z}\|z_{0}-\bar{z}_{0}\|+\rho(\varpi)\\
    &\quad+\sum_{\tau=0}^{k}(\rho_{w}\|w_{\tau}\|+\rho_{\varepsilon}\|\varepsilon_{y,\tau}\|+\rho_{d}\|d_{\tau}\|)\xi^{k-\tau-1}\nonumber
\end{align}
where $\tilde{c}_{z}=c_{z}(1-\xi^{N})$, $\rho_{z}=\epsilon+\tilde{c}_{z}\sigma+\sqrt{\eta_{N}\lambda_{z}\bar{p}}$, $\rho_{\varepsilon}=(\sqrt{\eta_{N}\bar{r}}+\tilde{c}_{z}\sigma_{y})\xi^{1-N}+c_{h}$, $\rho_{d}=(\sqrt{\eta_{N}\bar{q}}+\tilde{c}_{z}\sigma_{z})\xi^{1-N}+2c_{h}$, and $\rho_{w}=\omega_{w}$, $\rho(\varpi)=\big(\tilde{c}_{z}\sigma_{\alpha}+(\frac{2c_{\alpha}(\sqrt{\eta_{N}\bar{q}}+\tilde{c}_{z}\sigma)}{(c_{z}+1)}+\frac{2c_{\alpha}(N+1)(\sqrt{\eta_{N}\bar{r}}+\tilde{c}_{z}\sigma_{y})}{c_{h}N})\varpi+\sqrt{\eta_{N}\bar{\lambda}_{\alpha}}\big)\bar{\alpha}$.
Select $\theta_{z} = \omega_{z}^{\frac{1}{2N}}$, which satisfies 
$\omega_{z}^{j}\leq\theta_{z}^{\tilde{k}+jN}$ for $\tilde{k}\in[0,N)$, $j\in\mathbb{N}_{[1,\infty)}$.
It follows that 
\begin{align}\label{eq:derive_w}
    \sum_{\tau=0}^{\tilde{k}}\|w_{\tilde{k}-\tau}\|\omega_{z}^{j}\!+\!\sum_{i=0}^{j-1}\sum_{\tau=0}^{N}\|w_{\tilde{k}+(j-i)N-\tau}\|\omega_{z}^{i}\!\leq\!\tilde{\rho}\sum_{\tau=0}^{k}\|w_{k-\tau}\|\theta_{z}^{\tau}
\end{align}
where $\tilde{\rho}=\max\{2,\theta_{z}^{1-N}\}$. By substituting \eqref{eq:38} into \eqref{eq:delta xk-j(N+1)} and applying the procedure of \eqref{eq:derive_w} to the terms involving $\varepsilon_{y,k}$ and $d_{k}$, we obtain 
\begin{align*}
     \|z_{k}-\hat{z}_{k|k}^{*}\|&\leq \rho_{z}\|z_{0}-\bar{z}_{0}\|\theta_{z}^{k}+\theta(\varpi)\\
    &+\sum_{\tau=0}^{k}(\gamma_{w}\|w_{k-\tau}\|+\gamma_{\varepsilon}\|\varepsilon_{y,k-\tau}\|+\gamma_{d}\|d_{k-\tau}\|)\theta_{z}^{\tau}\nonumber
\end{align*}
where $\gamma_{w}=\rho_{w}\xi^{-1}\tilde{\rho}$, $\gamma_{\varepsilon}=\rho_{\varepsilon}\xi^{-1}\tilde{\rho}$, $\gamma_{d}=\rho_{d}\xi^{-1}\tilde{\rho}$, and  $\theta(\varpi)=\rho(\varpi)+\frac{\omega(\varphi)}{1-\omega_{z}}$. 
By \eqref{eq:surrogate_3}, we further derive that
\begin{align*}
     \|x_{k}-\hat{x}_{k|k}^{*}\|
     &\leq \|D\|(\rho_{z}L_{\psi}\|x_{0}-\bar{x}_{0}\|\theta_{z}^{k}+\sum_{\tau=0}^{k}(\gamma_{w}\|w_{k-\tau}\|\\
    &\quad+\gamma_{\varepsilon}\|\varepsilon_{y,k-\tau}\|+\gamma_{d}\|d_{k-\tau}\|)\theta_{z}^{\tau}+\theta(\varpi))+\bar{r}\nonumber
\end{align*}
This completes the proof. $\square$
\end{pf}

{\color{black}Theorem \ref{thm:stability} analyzes the robustness of the proposed data-enabled MHE against measurement noise; the analysis shows that the estimation error of the proposed method remains bounded even in the presence of measurement noise. At the same time, the error bound in \eqref{eq:thm_pres} depends on the magnitude of the measurement noise, which indicates that higher noise levels may degrade estimation performance.}

\section{Applications to membrane bioreactor}
{\color{black}In this section, a benchmark simulation model for membrane bioreactors (BSM-MBR) (\cite{maere2011bsm}) is utilized to evaluate the performance of the proposed learning-based data-enabled MHE method. }

\subsection{Process description}
We consider a membrane bioreactor used for wastewater treatment, with a schematic illustration provided in Fig. \ref{Fig.MBR} (\cite{maere2011bsm}). This plant comprises two anoxic tanks and three aerobic tanks, with the final aerobic tank being equipped with a membrane module for filtration. In the anoxic tanks, predenitrification occurs to convert nitrate into gaseous nitrogen. In the aerobic tanks, the ammonium is oxidized into nitrite through nitrification. The membrane bioreactor integrates membrane modules within the fifth tank for filtration (\cite{maere2011bsm}).

Two streams are fed into the first anoxic tank: 1) the influent wastewater at concentration $Z_{f}$ and flow rate $Q_{f}$; 2) a recycle stream from the outlet of the second aerobic tank at concentration $Z_{r2}$ and flow rate $Q_{r2}$.
The outlet of the membrane bioreactor process consists of three parts: 1) the permeate, which contains purified water through membrane filtration, continuously withdrawn at concentration $Z_{e}$ and flow rate $Q_{e}$; 2) a sludge recycle stream, returned from the third aerobic tank to the first aerobic tank with concentration $Z_{r1}$ and flow rate $Q_{r1}$; 3) the waste sludge stream, discharged at concentration $Z_{w}$ and flow rate $Q_{w}$ (\cite{maere2011bsm}).

This process comprises eight biological reactions, and it involves thirteen major substances. The concentrations of these substances in five biological reactors constitute the 65 state variables of this process.
The definitions of these state variables and a detailed description of the membrane bioreactor process can be found in \cite{maere2011bsm}. To monitor the process, eight sensors are assumed to be installed in each biological tank, which can be found in Table~3 in \cite{li2023partition}. The control inputs of the membrane bioreactor process comprise the airflow rate $Q_{a,i}$, $i\in\mathbb{N}_{[1,3]}$, in three aerobic tanks and the internal recycle flow rate $Q_{r2}$. Additionally, this system is influenced by uncontrollable inputs, including inlet flow rate $Q_{f}$ and the concentration $Z_{f}$ of thirteen inlet substances. The corresponding influent data profiles are obtained from the International Water Association website\footnote{http://www.benchmarkwwtp.org} (\cite{alex2008benchmark}).
{\color{black}The sampling period of this process is set as 15 min.}

\subsection{Simulation setting}
{\color{black}
To learn the lifting functions using neural networks, we first conduct open-loop simulations of the first-principles model in \cite{guo2020nonlinear} to generate datasets containing input, output, and state trajectories. 
During the offline stage, a dataset of length $T=53760$ is generated and partitioned into a training set with 43008 samples and a validation set with 10752 samples. In addition, a dataset of length $T=1344$ for an operating period of 14 days is generated for testing, and a dataset of length $T=5000$ is used for constructing Hankel matrices.
For online implementation of the proposed data-enabled MHE, a dataset of length $T=1344$ is generated for constructing Hankel matrices, while another dataset of length $T=672$ is used for evaluating the performance of the proposed method.}
Note that the five state variables (i.e., $S_{I}$ in each tank) remain invariant throughout the entire simulation; these states are excluded from the datasets.

To construct the datasets, the initial condition is set to a steady-state operating point of the membrane bioreactor process, as presented in Table~3 of \cite{guo2020nonlinear}.
The control inputs are randomly generated from a uniform distribution between $u_{\min}$ and $u_{\max}$, with $u_{\min}=[4000~\mathrm{Nm^{3}/h},~2000~\mathrm{Nm^{3}/h},~5000~\mathrm{Nm^{3}/h},~0~\mathrm{m^{3}/day}]^{\top}$ and $u_{\max}=[4500~\mathrm{Nm^{3}/h},~2500~\mathrm{Nm^{3}/h},~38012~\mathrm{Nm^{3}/h}$, $160900~\mathrm{m^{3}/day}]^{\top}$. Each sampled control input is kept constant for 40 sampling periods.
Additive unknown noise is generated from a Gaussian distribution with zero mean and standard deviation of $0.01\times u_{\max}$ and is added to the control inputs. 
{\color{black}Additive measurement noise is introduced into the offline datasets. Specifically, the state and measured output are contaminated by noise samples drawn from Gaussian distributions with zero mean and standard deviation of $0.01\times x_{0}$ and $0.01\times y_{0}$, respectively. The output measurement noise in the online dataset is generated in the same way and is added to the measured outputs accordingly.
}

\begin{figure}[tttt]
  \centering
  \includegraphics[width=0.48\textwidth]{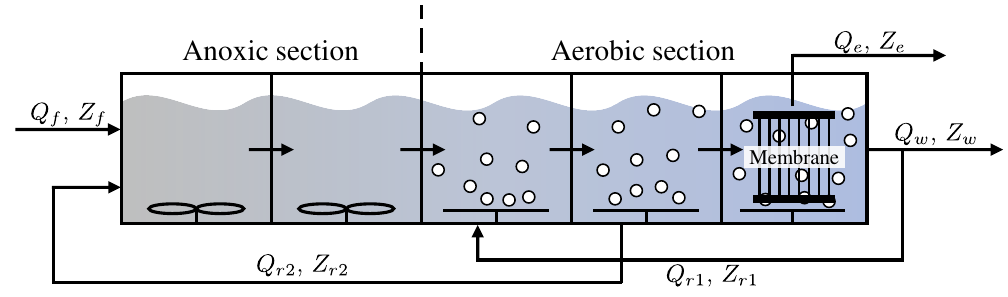}
  \caption{A schematic of the membrane bioreactor for wastewater treatment (\cite{maere2011bsm}).}\label{Fig.MBR}
\end{figure}

{\color{black}To implement the proposed learning-based data-enabled MHE approach, the dimensions of the scheduling parameter and the lifted states are chosen as $n_{p}=1$ and $n_{z}=80$, respectively. Two dense neural networks (\cite{rumelhart1986learning}), $\psi_{\theta}$ and $\lambda_{\sigma}$, are employed to approximate the state lifting function and the scheduling mapping, respectively. 
}
Specifically, $\psi_{\theta}$ has two hidden layers comprising 128 and 256 neurons, respectively, while $\lambda_{\sigma}$ consists of three hidden layers with 128, 256, and 256 neurons, respectively.
For both networks, the rectified linear unit (ReLU) (\cite{goodfellow2016deep}) is utilized as the activation function after the input and hidden layers. {\color{black}The hyperparameters used for training are shown in Table \ref{tbl:hyperparameters}.}
\begin{table}[t]
  \centering
{\color{black}
  \caption{Hyperparameters of the proposed method.}
  \label{tbl:hyperparameters}
  \small
  \renewcommand{\arraystretch}{1.1}
  \begin{tabular}{lc}
    \toprule[1pt]
    Hyperparameter & Value \\
    \midrule
    Training epochs & 400 \\
    Batch size & 256 \\
    Learning rate & $1\times10^{-4}$ \\
    Activation function & ReLU\\
    Optimizer & Adam \\
    Structure of $\psi_{\theta}$ & (128, 256) \\
    Structure of $\lambda_{\sigma}$ & (128, 256, 256) \\
    Scheduling parameter dimension $n_{p}$ & 1\\
    Lifted state dimension $n_{z}$ & 80\\
    \bottomrule[1pt]
  \end{tabular}
  \renewcommand{\arraystretch}{1.0}}
\end{table}

{\color{black}The neural network parameters (i.e., $\theta$ and $\sigma$) and the reconstruction matrix $D$ are jointly optimized by minimizing the 
loss function $\mathcal{L}$ defined in \eqref{eq:cost}. The training is performed over 400 epochs with a batch size of 256 using the Adam optimizer (\cite{kingma2014adam}) with a learning rate of $10^{-4}$.}

The initial guess $\bar{x}_{0}$ is selected as $1.05\times x_{0}$. The tuning parameters are chosen as $\lambda_{z}=600$, $\lambda_{\alpha}=1\times 10^{4}$, $\bar{\varepsilon}_{x}=\bar{\varepsilon}_{y}=0.001$, and $\bar{\Delta}_{z}=\bar{\Delta}_{y}=0.003$. The weighting matrices are set as $P = 100\times I_{80}$, $Q = 100\times I_{80}$, and $R = 1000\times I_{40}$.

\begin{figure}[tttt]
  \centering
  \includegraphics[width=0.48\textwidth]{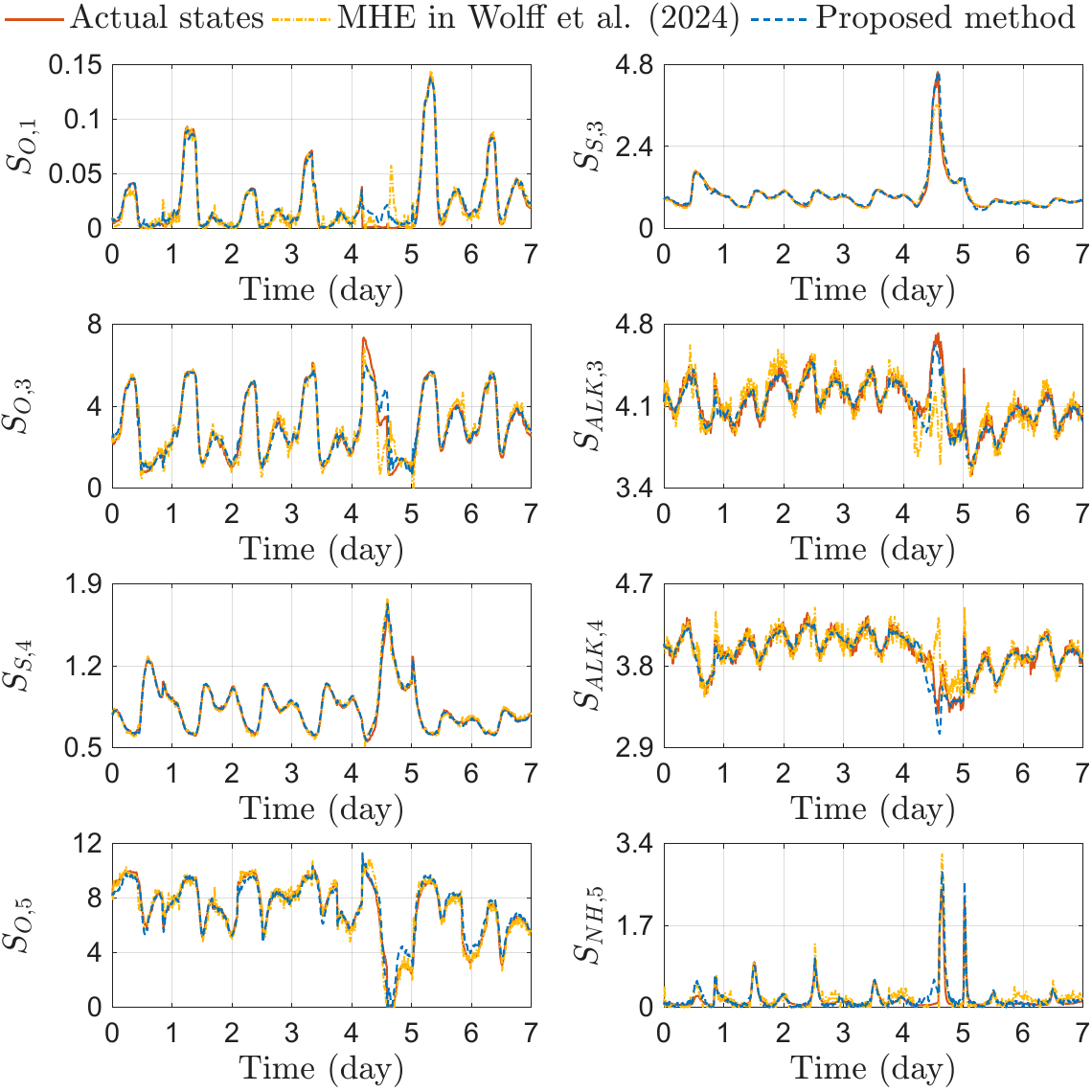}
  \caption{Trajectories of selected actual states and state estimates provided by the proposed data-enabled MHE in \eqref{eq:mhe} and MHE in \cite{wolff2024robust} under dry weather.}\label{Fig.MHE}
\end{figure}

\subsection{Estimation results}
{\color{black} The performance of the proposed method is evaluated using the average root mean square error (ARMSE) of the estimation error in scaled coordinates, as well as the computational time under dry weather conditions. These metrics are evaluated under different estimation horizons $N$, $N=3,\ldots,7$. 
For each estimation horizon, 10 Monte Carlo trials are performed with randomly generated initial state $x_{0}$ and measurement noise to obtain the ARMSEs and computation time. The results averaged over the 10 trials are presented in Table \ref{tbl:comparison_results}.}
We also evaluate the estimation performance of the data-enabled MHE method in \cite{wolff2024robust}, using identical parameters for fair comparison. As shown in Table \ref{tbl:comparison_results}, the proposed method outperforms the MHE approach in \cite{wolff2024robust} in terms of estimation accuracy.
{\color{black}
We note that the proposed method has a higher computation time than the data-enabled MHE method in \cite{wolff2024robust}. This is because the proposed approach applies Willems’ fundamental lemma in the lifted space, which leads to a large-scale optimization problem and an increased computational burden for solving the resulting optimization problem. However, the computation time of the proposed method remains relatively small as compared with the sampling period of the membrane-based biological wastewater treatment process, which is 15 minutes. 
This observation further indicates the applicability of the proposed method for deployment.

The estimation horizon is selected as $N=5$, as it achieves the lowest estimation error while maintaining a reasonable computational burden. A representative comparison of the estimation performance of the proposed data-enabled MHE method and the method in \cite{wolff2024robust} for $N=5$ is presented in Fig. \ref{Fig.MHE}. }

\begin{table}[tttt]
\centering
{\color{black}
\small
\caption{ARMSEs and computation time (s) provided by the proposed data-enabled MHE in \eqref{eq:mhe} and MHE in \cite{wolff2024robust} over 10 trials for different $N$.}
\label{tbl:comparison_results}
\renewcommand{\arraystretch}{1.2}
\begin{tabular}{c|cc|cc}
\Xhline{1px}
\multirow{2}{*}{$N$} & \multicolumn{2}{c|}{ARMSE} & \multicolumn{2}{c}{Computation Time} \\
\cline{2-5}
 & Proposed & Wolff et al. & Proposed & Wolff et al. \\
\hline
3 & 0.2925 & 0.3702 & 3.0027 & 0.8859 \\
4 & 0.2244 & 0.3247 & 4.6446 & 1.3781 \\
5 & 0.1802 & 0.2726 & 6.5830 & 1.9186 \\
6 & 0.1869 & 0.2452 & 9.2886 & 2.6774 \\
7 & 0.2079 & 0.2936 & 11.4764 & 3.5901 \\
\Xhline{1px}
\end{tabular}}
\end{table}

{\color{black}
\begin{rmk}
The trained neural networks used to approximate the lifting function and the scheduling mapping are trained on the training dataset, while they are evaluated on both the validation and test datasets. 
In addition, during the online implementation stage, both the offline dataset used for constructing the Hankel matrices and the online data used for estimation are generated under random input sequences that differ from those used in training. The proposed method maintains good estimation performance, which further demonstrates the generalization capability of the learned lifting function and scheduling mapping.
\end{rmk}
}

{\color{black}
\subsection{Sensitivity analysis}
To evaluate the robustness of the proposed data-enabled MHE approach with respect to measurement noise and disturbances, we conduct a sensitivity analysis by comparing the ARMSEs of the estimation error computed under scaled coordinates obtained by the proposed method. The results are calculated under different levels of measurement noise and disturbances. 
Specifically, we use a parameter $\mu$ to vary the standard deviations of measurement noise and disturbances. The disturbances follow a Gaussian distribution with zero mean and standard deviation of $\mu \times x_{0}$, where $x_{0}$ is the initial state. 
Similarly, the additive measurement noise, which is sampled from the Gaussian distributions with zero mean and standard deviation of $\mu \times x_{0}$ and $\mu \times y_{0}$, is introduced to the state and measured output, respectively.

The resulting  ARMSEs under different values of $\mu$ are presented in Table \ref{tbl:rmse}. The results show that the estimation errors remain low across different noise levels, which indicates the robustness of the proposed method to measurement noise and disturbances. We note that the proposed method remains robust, although such disturbances are not explicitly considered in the estimator design.}

\begin{table}[ttt]
\centering
{\color{black}
\small
\caption{ARMSEs provided by the proposed data-enabled MHE in \eqref{eq:mhe} for different levels of measurement noise and disturbances (i.e., different values of $\mu$).}\label{tbl:rmse}
\begin{tabular}{ccccccc}
\toprule[1pt]
$\mu$ & $10^{-5}$ & $10^{-4}$ & $10^{-3}$ & $ 10^{-2}$  \\
\midrule
Proposed method  & 0.2293  & 0.1849 & 0.2726  & 0.3234 \\
\bottomrule[1pt]
\end{tabular}}
\end{table}

\section{Conclusion}
We proposed a learning-based data-enabled MHE approach for {\color{black}the} nonlinear systems. This approach leverages an LPV Koopman surrogate of the underlying nonlinear system and employs two neural networks to construct the trajectories of this Koopman surrogate directly from system data. A data-enabled MHE was then formulated, which enables the reconstruction of the original nonlinear system states from the state estimates of the Koopman surrogate. The proposed approach does not require explicit system identification and formulates a convex optimization-based MHE design for {\color{black}the} nonlinear systems.
The stability of the proposed method was analyzed. The proposed method provides more accurate estimates than the benchmark data-enabled MHE for a membrane-based biological wastewater treatment process.

{\color{black}
Despite these advantages, the framework still has several limitations. In particular, its theoretical guarantees depend on the controllability and detectability of the Koopman surrogate. However, these conditions cannot be directly verified in practice, because the data-enabled formulation does not require explicit system matrices.
This limitation highlights an important direction for future research, namely, the development of verifiable conditions within the data-enabled framework. In addition, future work will focus on extending the proposed method to more complex wastewater treatment processes to enable real-time estimation of specific contaminant concentrations that are difficult to measure directly.
}

\end{document}